\documentclass[11pt]{article}
\usepackage[top=1in,bottom=1in,right=.8in,left=.8in]{geometry}
\usepackage{amssymb,amsmath,amsthm,xcolor}
\definecolor{forestgreen}{rgb}{0.13, 0.55, 0.13}
\usepackage[colorlinks, linkcolor=blue,citecolor=forestgreen,urlcolor = black, bookmarks=true]{hyperref}
\usepackage[T1]{fontenc}
 
\usepackage[nameinlink, noabbrev, capitalize]{cleveref}
\usepackage{tikz}
\usepackage{arydshln}
\usepackage{verbatim}
\usepackage{titlesec}
\usepackage{subfigure}
\usepackage[linesnumbered,ruled]{algorithm2e}
\usepackage{bm}
\usepackage{commath}
\allowdisplaybreaks

\titleclass{\subsubsubsection}{straight}[\subsection]

\newcounter{subsubsubsection}[subsubsection]
\renewcommand\thesubsubsubsection{\thesubsubsection.\arabic{subsubsubsection}}

\titleformat{\subsubsubsection}
  {\normalfont\normalsize\bfseries}{\thesubsubsubsection}{1em}{}
  \titlespacing*{\subsubsubsection}
  {0pt}{3.25ex plus 1ex minus .2ex}{1.5ex plus .2ex}

\makeatletter  
\def\toclevel@subsubsubsection{4}
\def\l@subsubsubsection{\@dottedtocline{4}{7em}{4em}}

\makeatother
\crefname{equation}{}{}
\crefrangeformat{equation}{(#3#1#4) --~(#5#2#6)}
\crefmultiformat{equation}{(#2#1#3)}{, (#2#1#3)}{, (#2#1#3)}{, (#2#1#3)}
\crefname{lemma}{Lemma}{Lemmas}
\crefname{section}{Section}{Sections}
\crefname{subsubsubsection}{Section}{Sections}
\crefname{remark}{Remark}{Remarks}
\crefname{figure}{Figure}{Figures}
\crefname{table}{Table}{Tables}
\Crefname{lemma}{Lemma}{Lemmas}
\crefname{theorem}{Theorem}{Theorems}
\Crefname{theorem}{Theorem}{Theorems}

\newtheorem{theorem}{Theorem}[section]

\newtheorem{remark}[theorem]{Remark}

\newtheorem{lemma}[theorem]{Lemma}

\newtheorem{proposition}[theorem]{Proposition}

\crefname{THM}{Theorem}{Theorems}

\theoremstyle{definition}

\theoremstyle{definition}
\newtheorem{definition}[theorem]{Definition}

\title{Polarization in Geometric Opinion Dynamics}

\author{Jason Gaitonde\thanks{Supported by NSF grant CCF-1408673 and AFOSR grant FA9550-19-1-0183.} \\
Cornell University\\
\texttt{jsg355@cornell.edu}
\and
Jon Kleinberg\thanks{Supported in part by a Simons Investigator Award, a
Vannevar Bush Faculty Fellowship, a MURI grant, a MacArthur Foundation grant, and AFOSR grant FA9550-19-1-0183.} \\
Cornell University\\
\texttt{kleinber@cs.cornell.edu}
 \and
\'Eva Tardos\thanks{Supported in part by NSF grants CCF-1408673, CCF-1563714, and AFOSR grant FA9550-19-1-0183.} \\
Cornell University\\
\texttt{eva.tardos@cornell.edu}}

\begin{document}

\maketitle
\begin{abstract}
In light of increasing recent attention to political polarization, understanding how polarization can arise poses an important theoretical question. While more classical models of opinion dynamics seem poorly equipped to study this phenomenon, a recent novel approach by H\k{a}z\l{}a, Jin, Mossel, and Ramnarayan (HJMR) \cite{DBLP:journals/corr/abs-1910-05274} proposes a simple \emph{geometric} model of opinion evolution that provably exhibits strong polarization in specialized cases. Moreover, polarization arises quite organically in their model: in each time step, each agent updates opinions according to their correlation/response with an issue drawn at random. However, their techniques do not seem to extend beyond a set of special cases they identify, which benefit from fragile symmetry or contractiveness assumptions, leaving open how general this phenomenon really is.
    
In this paper, we further the study of polarization in related geometric models. We show that the exact form of polarization in such models is quite nuanced: even when strong polarization does not hold, it is possible for \emph{weaker} notions of polarization to nonetheless attain. We provide a concrete example where weak polarization holds, but strong polarization provably fails. However, we  show that strong polarization  provably holds in many
variants of the HJMR model, which are also robust to a wider array of distributions of random issues---this suggests that the form of polarization introduced by HJMR is more universal than suggested by their special cases. We also show that the
weaker notions connect more readily to the theory of Markov chains on general state spaces.
\end{abstract}

\section{Introduction}

People's opinions naturally evolve in response to a variety of external factors, like interactions with other individuals, campaign messaging, media reports, political identities, and random events. Understanding how opinions evolve and the emergent qualitative features of these dynamics remains a subject of immense interest in the computer science, economics, and social science communities.

In recent years, a crucial phenomenon of interest is that of \emph{polarization}, where agents roughly partition into groups holding diametric views. That is, rather than agents holding a rich spectrum of beliefs, individuals instead typically belong to opposite clusters even in cases where beliefs on separate topics ostensibly ought not be correlated. This phenomenon is somewhat widely observed to have accelerated in the last couple of decades; for instance, Gentzkow, Shapiro, and Taddy show that political partisanship is more easily inferrable in recent years than previously from textual analysis using machine learning methods \cite{gentzkow2016measuring}. While polarization and issue realignment is perhaps most familiarly observed along political dimensions (for instance, the correlation between beliefs on climate change and gun rights), these effects are not limited to just these facets \cite{dellaposta2015liberals}.

A prototypical way to model the evolution of the beliefs of agents, and thereby address phenomena like polarization, is to assume that the opinion of each agent is scalar- or vector-valued and evolves according to some prescribed update rule. Many classical models, like the DeGroot and Friedkin-Johnsen models \cite{degroot1974reaching,friedkin1990social}, study opinion evolution as a form of social learning where agents update beliefs based on the graph-weighted opinions of their neighbors within an ambient network. However, the original forms of these models do not appear to provide a satisfactory mechanism to explain polarization, as the dynamics provably lead to \emph{less} polarization than in the initial configuration (see \Cref{sec:related} for a larger discussion).

To understand how polarization can arise in a more intrinsic way, H\k{a}z\l{}a, Jin, Mossel, and Ramnarayan (HJMR) recently introduced a novel and simple geometric model of opinion dynamics \cite{DBLP:journals/corr/abs-1910-05274}. In their model, each agent's opinion is a unit length vector where each dimension represents a particular relevant political axis. Their dynamics work as follows: at each time, a new random direction representing a new issue, political figure, or influencer is drawn from some fixed distribution. In response, each agent evaluates the correlation between their current opinions and this new issue and moves either toward or away from this new direction depending on the strength and sign of this correlation. The resulting vector is then renormalized to unit length (see \Cref{sec:preliminaries} for the exact model). These random directions model the natural intuition that opinion evolution is often driven by concrete, possibly random events. For instance, opinions potentially change dramatically in response to campaign messages or candidates during election seasons, when new political coalitions form and ideologies can themselves prove malleable to accommodate the new compositions of these groups. Because each vector is normalized, polarization has a very clean geometric interpretation: a set of opinion vectors is polarized if they are all equal up to sign. 

One of the key findings of HJMR is a proof that these simple dynamics \emph{intrinsically and strongly polarize} when restricted to a set of specialized settings. Concretely, when opinions are two-dimensional (so lie on the unit circle) and the new issues are drawn uniformly at random from the circle, or when vectors have arbitrary dimension but new issues correspond to one of two ``duelling influencers'' that are sufficiently close and drawn with equal probability, then almost surely the distance between any two starting opinions will converge to 0 or 2 via these random dynamics; in other words, almost surely any two agents will have opinions either converging to each other or to the negative of each other so that they are equal up to sign.

However, it is not at all clear whether or not this almost sure polarization is specific to their exact model formulation: for instance, their technique to show strong polarization in the two-dimensional case with uniform random directions does not obviously generalize to higher dimensions, nor to non-uniform distributions---without this high degree of symmetry, it is not clear whether the dynamics still model polarization (see \Cref{sec:overview} for more discussion).  Moreover, their dynamics are completely \emph{oblivious} in the sense that each agents' opinions evolve randomly in response to (the same) new issues but otherwise have no social influence over each other. A natural and important question in this framework is thus whether, or to what extent, polarization can occur more generally in related variants that do not satisfy these stringent constraints.

To address this, we can abstractly define geometric opinion dynamics more generally as follows: the key components of the HJMR model are that (i) the dynamics are random to model the impact of new political issues that arise, and in fact form a Markov chain on the hypersphere, and (ii) the set of polarized vectors is invariant under the dynamics. This latter property is of course necessary to prove that polarization occurs, while the former provides a convenient mathematical framework to analyze these dynamics. This motivates the following generalization of their dynamics:

\begin{definition}
Let $n\in \mathbb{N}$ be an arbitrary, but otherwise fixed parameter denoting the number of agents, and let $d$ denote the dimension of opinions. In a general model of \textbf{geometric opinion dynamics}, the evolution of each agent's opinion follows a discrete-time Markov chain in $\mathbb{S}^{d-1}$ given by a recurrence of the following form:
\begin{equation}
\label{eq:dynamics}
    \mathbf{X}_{t+1}^{(i)} \propto \mathbf{X}_t^{(i)} + f_i(\mathbf{X}_{t}^{(1)},\ldots,\mathbf{X}_{t}^{(n)},\xi_t), 
\end{equation}
where $\mathbf{X}_0=(\mathbf{X}_{0}^{(1)},\ldots,\mathbf{X}_{0}^{(n)})\in (\mathbb{S}^{d-1})^n$ is a given starting configuration of opinions, $\xi_t\in \mathbb{S}^{d-1}$ is drawn i.i.d. over time (and independent of all other random variables) from some distribution $\mathcal{D}$, and each $f_i:(\mathbb{S}^{d-1})^{n+1}\to \mathbb{R}^d$ is an explicit, fixed updating function.
\end{definition} 
The implicit normalization in this formulation ensures that each opinion returns to the unit sphere; in particular, we assume that the unnormalized quantity above is nonzero so that the projection onto the sphere is well-defined. The joint dynamics of this system thus take place as a Markov chain in $\prod_{i=1}^n \mathbb{S}^{d-1}$. In such models, the random vector $\xi_t$ again denotes a common, random stimulus that each agent in the system updates opinions with respect to, but now possibly also as a function of the opinions of the other agents in the system. For example, $\xi_t$ may model a particular issue or figure during election season that splits society into two camps, ``for'' and ``against,'' and opinions update according to these coalitions. Crucially, we also assume that the distribution $\mathcal{D}$ that these vectors are drawn from remains fixed throughout time so that these dynamics form a time-homogeneous Markov chain.

With this definition, we extend the results of HJMR in this paper by pursuing a more systematic study of polarization in these general geometric opinion dynamics models that take the form of \Cref{eq:dynamics}. The motivating questions we consider are: \emph{does strong polarization hold more generally, or is it specific to the HJMR model formulation? Can strong polarization be shown with nontrivial network interactions, thereby incorporating a key feature of models like the DeGroot and Friedkin-Johnsen dynamics? Are there other interesting notions of polarization that hold in these settings, even if the strong form does not hold?}

We show that whether polarization holds in such models is rather surprisingly nuanced and requires extending the notion of polarization beyond what is proven in HJMR. In particular, we show that there exist models that \emph{weakly} polarize in a formal sense, but provably do not satisfy the stronger form in HJMR. We show that these weaker forms of polarization are connected to the existence of nontrivial invariant distributions on the induced Markov chain.  Nonetheless, we also prove that strong polarization holds for nontrivial variants of the HJMR model which are also robust to more general distributions of update vectors---in fact, we show that this holds in a model that has network interactions in addition to the random updates. We hope that some of our techniques will prove useful in studying polarization in such models more generally.

\subsection{Overview of Results and Techniques}
\label{sec:overview}
We contribute to the theoretical understanding of polarization in geometric opinion models initiated in \cite{DBLP:journals/corr/abs-1910-05274}. In \Cref{sec:preliminaries}, we begin by identifying, though not necessarily requiring, several natural natural properties that such dynamics might satisfy and their relation to the problem of polarization. These will be convenient in our later arguments. We also define formally several distinct geometric opinion models, including the HJMR model, that will be the subject of our analysis later on.

Our main results in \Cref{sec:almostsure} show that \emph{strong polarization is more universal in such models than perhaps it would appear from the specialized cases studied by HJMR}. In particular, we show that the strong form of polarization exhibited by HJMR in restricted versions of their model holds in two other variants we consider. The first model we study, which we term the \emph{signed HJMR model}, considers variant of the HJMR dynamics that replaces an inner product with a sign. With these altered dynamics, we show that strong polarization holds in any dimension and is robust to the choice of distribution on update vectors so long as it is sufficiently close to the uniform measure on the sphere. We also show that the same holds in the \emph{party model}, which notably incorporates \emph{network effects} where agents exert influence over each other---such effects are not present in HJMR and it is not clear how to apply their methods to this setting.

To prove our results in \Cref{sec:almostsure}, we resort to conceptually different techniques from those of HJMR. Their work shows strong polarization primarily by appealing to either martingale convergence (in the case of $d=2$ and with a uniform distribution), which relies on rather delicate symmetries, or deterministic contraction (in two-point models representing ``duelling influences'', where the challenge of proving contraction is geometric rather than probabilistic). To establish strong polarization in these variants, we appeal to more general zero-one laws to first simplify the problem to showing just that there is \emph{some} nonzero probability of strongly polarizing. We obtain our main results from this quite general reduction by employing clean probabilistic and geometric estimates to give the required uniform bounds in our analysis to prove strong polarization.

Our next main contribution is to connect the polarization phenomenon to the general theory of Markov chains in uncountable state spaces in \Cref{sec:markov}. In general, Markov chains in such spaces have considerably more delicate behavior than with countable or finite state spaces. To do so, we define weaker notions of polarization than the strong form considered previously; these notions are motivated by standard notions of convergence from probability theory. We first show that in most cases, whether these weaker forms of polarization hold in geometric opinion dynamics is in fact equivalent to the existence of nontrivial invariant distributions of the Markov chain defined by the dynamics. We then turn to showing the utility of these more nuanced notions of polarization with a concrete example: we conclude by proving that the original HJMR model with orthonormal updates satisfies our notion of weak polarization, but provably \emph{does not} satisfy strong polarization---in fact, we show that almost surely, almost every starting configuration will not strongly polarize. We do so by connecting these dynamics to pathwise properties of infinite balls-in-bins processes, which we analyze by applying yet another zero-one law of Hewitt and Savage.

\subsection{Related Work}
\label{sec:related}
Our work broadly relates to the vast body of theoretical work in opinion dynamics that studies concrete models wherein agents update numerical beliefs and the resulting properties of these dynamics; for a comprehensive survey of this style of work, see \cite{castellano2009statistical}.
The particular geometric model of opinion dynamics considered by H\k{a}z\l{}a, Jin, Mossel, and Ramnarayan \cite{DBLP:journals/corr/abs-1910-05274} motivates many of the considerations in this work. To our knowledge, their work is first to introduce the dynamics in \Cref{eq:dynamics} in the specific form given below in \Cref{eq:hjmr}. We are not aware of related work exploring polarization in variants of their model.

Polarization more generally has been studied in other models of opinion dynamics, typically in the more well-understood DeGroot and Friedkin-Johnsen dynamics mentioned above. Both models assume that agents update opinions according to the graph-weighted average of their neighbors, and both of these models have been extensively studied for their analytical elegance, tractability, and deep connections to well-studied topics like finite-state Markov chains and random walks \cite{golub2010naive,DBLP:conf/sdm/GionisTT13}. However, an idiosyncratic feature of these models is that they are often inherently contractive: opinions get closer due to the dynamics---for instance, in the DeGroot model, agents will necessarily converge to a common limiting opinion under very mild connectedness assumptions. Therefore, the price of the analytic simplicity of this model is that no nontrivial polarization can occur. The Friedkin-Johnsen model typically will not exhibit perfect consensus, but the network effects of the process nonetheless cause initial opinions to contract closer to a common consensus. This suggests that these models simply are not well-equipped to offer a mechanism for this phenomenon. 

Recent theoretical work has attempted to combine the analytically desirable features of these models with questions of polarization by incorporating new elements to the model. For instance, behavioral biases like \emph{biased assimilation}, which underlies the intuition behind the HJMR model, have also been fruitfully studied in the context of DeGroot dynamics \cite{DBLP:journals/pnas/DandekarGL13}. Motivated by recent events, a related line of work has also attempted to understand polarization by incorporating the influence of external actors or platforms into the dynamics. Polarization then arises either explicitly or implicitly because of the (possibly orthogonal) objectives of these external parties \cite{DBLP:conf/www/MuscoMT18,DBLP:conf/wsdm/ChitraM20,DBLP:conf/sigecom/GaitondeKT20,immorlica}. A related model by Hegselmann and Krause \cite{hegselmann2002opinion} circumvents near-consensus by allowing agents to filter out overly dissimilar opinions before averaging at each step, but already this modest extension incurs a high cost: many basic questions about convergence remain open in this model, greatly complicating a more extensive understanding of how polarization arises \cite{DBLP:conf/innovations/BhattacharyyaBCN13}. Moreover, the polarization in this model is somewhat built-in by requiring that far opinions stop responding to each other. While these lines of work have shed significant light on how polarization can arise in such models, and certainly such forces can help explain an empirical increase in polarization in the last few decades, the polarization that arises does not arise organically via the original dynamics.

As we discuss in \Cref{sec:markov}, geometric models of opinion dynamics can be viewed as a special case of the more general theory of Markov chains in non-discrete state spaces. The results in this setting are significantly more involved than corresponding results on discrete state spaces; an accessible reference to some important results in this area can be found in the surveys of Hairer \cite{hairer2006ergodic,hairer2010convergence}. As we show, polarization is related to the set of invariant distributions on the Markov chain induced by the dynamics. Numerous techniques have been developed to determine the uniqueness of an invariant distribution in these settings, as well as the rate and mode of convergence to it if possible (see for instance, \cite{butkovsky2014subgeometric}). It would be interesting to find new ways to apply results from this area to provide quantitative bounds on the convergence to polarization when one can prove strong polarization holds.

\section{Preliminaries}
\label{sec:preliminaries}

As stated above, we are primarily interested in the \emph{polarization} properties of these random processes given by \Cref{eq:dynamics}. Throughout this paper, we reserve $n$ to denote the number of agents and $d$ to denote the dimensionality of opinions. We write $\|\cdot\|_p$ for the standard $\ell_p$-norm and  $\mathbb{S}^{d-1}=\{\mathbf{x}\in \mathbb{R}^d: \|\mathbf{x}\|_2=1\}$. We will write $P_{\mathbb{S}^{d-1}}:\mathbb{R}^d\setminus \{\mathbf{0}\}\to \mathbb{S}^{d-1}$ for the projection onto the unit sphere, i.e. $P_{\mathbb{S}^{d-1}}(\mathbf{x})=\mathbf{x}/\|\mathbf{x}\|_2$. We also write $\mathbf{e}_i$ for the $i$th standard basis vector in $\mathbb{R}^d$, where $d$ will be clear from context. For a set of vectors $\mathbf{x}_1,\ldots,\mathbf{x}_n\in \mathbb{R}^d$, we define $\text{cone}(\mathbf{x}_1,\ldots,\mathbf{x}_n)=\{\mathbf{z}\in \mathbb{R}^d: \mathbf{z}=\sum_{i=1}^n \alpha_i \mathbf{x}_i, \alpha_i\geq 0$\}. We write $\angle (\mathbf{x},\mathbf{y})$ for the angle between $\mathbf{x}$ and $\mathbf{y}$.

For given $n,d$, we then define $D \subseteq \prod_{i=1}^n \mathbb{S}^{d-1}$ to be the set of diagonal opinion vectors, i.e. the set of elements of the form $(\mathbf{x},\ldots,\mathbf{x})$ for a single vector $\mathbf{x}\in \mathbb{S}^{d-1}$. For a sign vector $\bm{\sigma}\in \{-1,1\}^n$, define $\bm{\sigma}(D):= \{\mathbf{X}\in \prod_{i=1}^n \mathbb{S}^{d-1}: \mathbf{X} = (\sigma_1\mathbf{x},\ldots,\sigma_n\mathbf{x}), \mathbf{x}\in \mathbb{S}^{d-1}\}$. Finally, we define the set $P$ of polarized vectors by $P = \bigcup_{\bm{\sigma}\in \{-1,1\}^n} \mathbf{\sigma}(D).$
In other words, $P$ is the set of tuples of vectors such that each vector is equal to the rest up to sign. This definition allows for consensus as a special case, but in many settings, consensus or near-consensus is exponentially unlikely (see \Cref{rmk:consensus}).

For a subset $A$ in some Euclidean space and a point $\mathbf{x}$, we define $\rho(\mathbf{x},A):=\inf_{\mathbf{y}\in A} \|\mathbf{x}-\mathbf{y}\|_2$ to denote the distance from a point to a set with respect to the Euclidean metric. With this in mind, we define strong polarization of a geometric opinion dynamics model as follows:
\begin{definition}
\label{def:strongpolar}
Let $\mathbf{X}_t:=(\mathbf{X}_t^{(1)},\ldots, \mathbf{X}_t^{(n)})$ be a discrete-time Markov chain as given by \cref{eq:dynamics}. Then $\mathbf{X}_t$ \textbf{strongly polarizes} (from $\mathbf{X}_0$) 
if almost surely, $\rho(\mathbf{X}_t,P)\to 0$; that is, the distance between $\mathbf{X}_t$ and the set of polarized vectors converges to zero almost surely.
\end{definition}

There are several natural properties that one might desire in these dynamics. Below, we consider, though do not require, the following properties which abstracts those of the original HJMR model:

\begin{definition}
If each function $f_i(\cdot,\xi)$ is continuous from $\prod_{i=1}^{n}\mathbb{S}^{d-1}$ to $\mathbb{R}^d$ for each fixed $\xi$, then we say the dynamics are \textbf{continuous}.
\end{definition} 
We will show in \Cref{sec:markov} that continuity implies various desirable properties for the dynamics. Note that if this is the case, combined with the fact that we assume the unnormalized update rule is always nonzero, it follows that the dynamics are jointly continuous as transitions from $\prod_{i=1}^n \mathbb{S}^{d-1}$ to itself. This follows because the the map taking the joint opinion vector to the unnormalized opinion vectors is continuous and nonzero in every coordinate, which is then composed coordinatewise with a continuous map on $\mathbb{R}^{d}\setminus \{\mathbf{0}\}$, and so is continuous.

\begin{definition}
The dynamics are \textbf{sign-invariant} if, for all $i\in [n]$, the function $f_i$ is \emph{odd} with respect to $\mathbf{x}_i$, but \emph{even} with respect to the other arguments (i.e. with respect to $\mathbf{x}_{-i}$ and $\xi$).
\end{definition}
Sign-invariance implies each agent reacts to the random update vector and the others the same regardless if any are negated. This feature will be present in all of the models we consider below, including the original HJMR model. The intuition behind sign-invariance is that from the perspective of each agent, if he or she were to react ``positively'' to the new issue or a different opinion, she would react ``negatively'' to the negative of that issue or different opinion---sign-invariance thus asserts that these reactions are balanced.
    
\begin{definition}
If $f_i(\mathbf{X}_t^{(i)},\mathbf{X}^{(-i)}_t,\xi) = f_j(\mathbf{X}_t^{(j)},\mathbf{X}^{(-j)}_t,\xi)$ when $\mathbf{X}_t^{(i)}=\mathbf{X}_t^{(j)}$ and $\mathbf{X}_t^{(-i)}=\mathbf{X}_t^{(-j)}$, we say that the dynamics are \textbf{symmetric}, as the updates do not depend on the identities of the agents.
\end{definition}
    
\begin{definition}
If each function $f_i$ does not depend on $\mathbf{X}_t^{(-i)}$ for each $i$ (so that $f_i$ depends only on $\mathbf{X}_t^{(i)}$ and $\xi_t$), then we say the dynamics are \textbf{oblivious}.
\end{definition} 
In this case, each component of the above process follows a Markov chain, and the joint dynamics form a particular coupling where each component responds to the same update vector. However, any polarization that arises happens indirectly because agents do not influence each other.

In this work, we treat such models in relatively full generality and also specialize to particular models where more specific techniques can establish various desirable properties. The concrete examples we will consider are listed below:

\begin{definition}[HJMR Model]
In the \textbf{HJMR Model} \cite{DBLP:journals/corr/abs-1910-05274}, the update for each agent $i$ takes the following form for some fixed scalar $\eta>0$:
    \begin{equation}
    \label{eq:hjmr}
        \mathbf{X}_{t+1}^{(i)} \propto \mathbf{X}_t^{(i)} + \eta\cdot \langle \mathbf{X}_t^{(i)},\xi_t\rangle \xi_t,
    \end{equation}
    where $\xi_t\sim \mathcal{D}$ is drawn i.i.d. over time from some distribution $\mathcal{D}$ on $\mathbb{S}^{d-1}$. In words, each agent moves in the (signed) direction of the random update vector proportionally to the correlation with their current opinion, and then renormalizes.
\end{definition}
    Note that this model satisfies continuity, sign-invariance, obliviousness, and symmetry (assuming $\eta$ is a constant over all agents).
\begin{definition}[Signed HJMR Model]
In the \textbf{signed HJMR model}, the update rule in \Cref{eq:hjmr} is amended to
    \begin{equation}
    \label{eq:hjmrsign}
        \mathbf{X}_{t+1}^{(i)} \propto \mathbf{X}_t^{(i)} + \eta\cdot  \text{sgn}(\langle \mathbf{X}_t^{(i)},\xi_t\rangle) \xi_t.
    \end{equation}
\end{definition}
Here, we define $\text{sgn}(0)=0$, but in our applications below, we will assume $\xi_t$ is drawn from a continuous distribution so that almost surely $\langle \mathbf{X}_t^{(i)},\xi_t\rangle\neq 0$. In this case, the amount the vector updates does not depend on the correlation. This choice is intended to model elections, where one is in favor either towards or against a particular candidate and is drawn ``all-or-nothing'' towards or against the views of this candidate. This model is \emph{not} continuous due to the sign function, but is still sign-invariant, oblivious, and symmetric (assuming $\eta$ is a constant over all agents).
    
\begin{definition}[Party Model]
Suppose each agent $i$ has multipliers $(\eta_1^{(i)},\ldots,\eta_n^{(i)})$ where $\eta^{(i)}_j\geq 0$ measures the influence of agent $j$ on agent $i$. The \textbf{party model} is defined by:
    \begin{equation}
    \label{eq:party}
        \mathbf{X}_{t+1}^{(i)} \propto \mathbf{X}_t^{(i)} + \left(\sum_{j\in [n]:\text{sgn}(\langle \mathbf{X}_t^{(j)},\xi_t\rangle)=\text{sgn}(\langle \mathbf{X}_t^{(i)},\xi_t\rangle)} \eta_j^{(i)}\mathbf{X}_t^{(j)}-\sum_{j\in [n]:\text{sgn}(\langle \mathbf{X}_t^{(j)},\xi_t\rangle)\neq\text{sgn}(\langle \mathbf{X}_t^{(i)},\xi_t\rangle)} \eta_j^{(i)}\mathbf{X}_t^{(j)}\right).
    \end{equation}
\end{definition}
    That is, each agent moves towards the vectors that came on the same ``side'' of the random issue $\xi_t$, and away from those on the opposite ``side'' of the issue. While this latter assumption may appear non-obvious, there is considerable empirical evidence for the sociological principle that ``out-group conflict builds in-group solidarity'' \cite{mccallion2007groups}: when a binary issue creates disagreement within a collection of people, the people on each side of the disagreement move toward those they agree with and away from those they disagree with \cite{fisher2016towards,sherif-harvey-white-hood-sherif61}. Once again, this model is not continuous due to the sign function, and also is \emph{not} oblivious as clearly the update rule depends on the values of the other agents. However, it remains sign-invariant, as negating one's own opinions interchanges the sums, and flipping either $\xi_t$ or any other agents opinions only permutes summands.

For any sort of polarization to arise, it is natural to ensure that the dynamics are such that \emph{if} the vector are completely polarized, then they will surely remain so. We provide one simple condition that is easily verified\footnote{Note that the converse of \Cref{lem:Pinv} trivially fails: one can simply ensure the dynamics are invariant on each such set and otherwise define them arbitrarily.}:

\begin{lemma}
\label{lem:Pinv}
Suppose that some geometric opinion dynamics satisfy symmetry and sign-invariance. Then $\bm{\sigma}(D)$ is invariant under the transitions for every $\bm{\sigma}\in \{-1,1\}^n$.
\end{lemma}
\iffalse
\begin{proof}[Proof of \Cref{lem:Pinv}]
Suppose that $\mathbf{X}_t$ is such that $\mathbf{X}_{t}^{(i)}=\sigma_i x$ for some $x\in \mathbb{S}^{d-1}$. By symmetry, we know unconditionally that if $\mathbf{X}_t^{(i)}=\mathbf{X}_t^{(j)}$, then $\mathbf{X}_{t+1}^{(i)}=\mathbf{X}_{t+1}^{(j)}$. Now, if $\mathbf{X}_t^{(i)}=-\mathbf{X}_t^{(j)}$, then we observe that
\begin{align*}
    \mathbf{X}_{t+1}^{(j)}&\propto \mathbf{X}_t^{(j)} + f(\mathbf{X}_{t}^{(j)},\mathbf{X}_{t}^{(-j)},\xi_t)\\
    &=-\mathbf{X}_t^{(i)} + f(-\mathbf{X}_{t}^{(i)},\mathbf{X}_{t}^{(-i)},\xi_t)\\
    &\propto -\mathbf{X}_{t+1}^{(i)}.
\end{align*}
The second line uses the fact that $\mathbf{X}_{t}^{(-j)}$ and $\mathbf{X}_{t}^{(-i)}$ differ only up to signs by assumption and then uses the evenness part of sign-invariance with respect to the other agents. The last line uses the oddness part of sign-invariance with respect to each agent's opinions.
\end{proof}
\fi
Finally, note that if the dynamics are oblivious and strongly polarizes for $n=2$ agents, then it must do so for any finite $n$ by a simple union bound. Note that oblivious and symmetric dynamics ensure that the process is well-defined for any number of agents.
\begin{lemma}
\label{lem:2suff}
Suppose that the opinion dynamics are symmetric and oblivious. Then if any form of convergence holds for $n=2$ agents, the same holds for any finite number of agents.

\end{lemma}

\section{Models with Strong Polarization}
\label{sec:almostsure}

In this section, we prove the strong polarization of the signed HJMR model and the party model. To do this, we first establish a simple, but powerful general principle that will significantly simplify the analysis that has the following simple intuition: suppose momentarily that geometric opinion dynamics were a \emph{finite-state} Markov chain and that ``polarization'' is an absorbing state of the Markov chain. Then from standard and simple Markov estimates, so long as it is possible to reach the state ``polarized'' from any starting point in some fixed finite number of steps, an easy calculation shows that almost surely the Markov chain will become ``polarized.'' In that case, it would suffice to show that from any starting state, there is \emph{some} nonzero probability of reaching ``polarized'' from any starting state in some finite number of steps.

In general, this idea is not so simple to formalize because the geometric opinion dynamics lie in a non-discrete state space, and moreover, it is often only possible to reach $P$ asymptotically, not in any finite number of steps. However, by appealing to more general zero-one laws, we show that this intuition nonetheless holds:

\begin{theorem}
\label{thm:levy}
For any geometric model of opinion dynamics that satisfies the Markov property, the following are equivalent:

\begin{enumerate}
    \item For every choice of starting vector $\mathbf{X}_0$, $\Pr_{\mathbf{X}_0}\left(\rho(\mathbf{X}_t,P)\to 0\right)=1.$
    
    \item For every choice of starting vector $\mathbf{X}_0$, $\Pr_{\mathbf{X}_0}\left(\rho(\mathbf{X}_t,P)\to 0\right)\geq c$
    for some constant $c>0$.
\end{enumerate}
\end{theorem}
\begin{proof}
One direction is trivial, so we assume the second condition. Consider the dynamics started at any choice of starting vector $\mathbf{X}_0$, and consider the event $A = \{\rho(\mathbf{X}_t,P)\to 0\}.$
Let $\mathcal{F}_t=\sigma(\xi_0,\ldots,\xi_{t-1})$ be the filtration generated by the random updates up to time $t$ and let $\mathcal{F}_{\infty}=\sigma(\xi_0,\ldots)$ be the filtration generated by all of them. By standard arguments, $A\in \mathcal{F}_{\infty}$.

For each $T\geq 0$, define $A_T:=\mathbb{E}[\bm{1}(A)\vert \mathcal{F}_{T}]$. As $\bm{1}(A)$ is an indicator random variable, L\'evy's upward theorem (Theorem 4.2.11 of \cite{durrett2019probability}) implies that $A_T=\mathbb{E}[\bm{1}(A)\vert \mathcal{F}_{T}]\to \mathbb{E}[\bm{1}(A)\vert \mathcal{F}_{\infty}]=\bm{1}(A)$ almost surely. But observe that by the Markov property, $A_T = \Pr_{\mathbf{X}_T}\left(\rho(\mathbf{Z}_t,P)\to 0\right)$, where $\mathbf{Z}_t$ gives the dynamics started at $\mathbf{Z}_0=\mathbf{X}_T$. In particular, $A_T\geq c>0$ pointwise. Because $A_T\to \bm{1}(A)$ almost surely and is bounded from below almost surely by a strictly positive quantity, the only way this can happen is if $\bm{1}(A)\equiv 1$ almost surely (over the realizations of the $\xi_t$), so that $A$ holds almost surely. As $\mathbf{X}_0$ was arbitrary, this completes the harder direction.
\end{proof}

\subsection{Strong Polarization in Signed HJMR Model}
Our first main result is that strong polarization holds in the signed HJMR model with a common value of $\eta>0$, and that this holds for a general class of distributions: 

\begin{theorem}
\label{thm:shjmr}
Suppose there are $n\geq 1$ agents in the signed HJMR model given by \Cref{eq:hjmrsign} where each $\xi_t$ is drawn i.i.d. from a distribution $\mathcal{D}$ that is equivalent to the uniform (Haar) measure $\mu$, i.e. there exists $C,C'>0$ such that for every measurable set $A\subseteq \mathbb{S}^{d-1}$, $C\mu(A)\leq \Pr(\xi_t\in A)\leq C'\mu(A)$. Then this system strongly polarizes from any choice of starting vector $\mathbf{X}_0$.
\end{theorem}

\begin{remark}
\label{rmk:consensus}
It can be shown that for any sign-invariant dynamics that strongly polarizes, if $\mathbf{X}_0$ is drawn uniformly from $(\mathbb{S}^{d-1})^n$, then each possible clustering is equally likely even conditioned on the sequence $\{\xi_t\}_{t=0}^{\infty}$ almost surely. In particular, the probability over starting configurations and the random updates of consensus is exponentially small in $n$.
\end{remark}

To set up this result, we establish a sequence of lemmas that will prove useful. We begin with a simple geometric fact:
\begin{lemma}
\label{lem:cont}
Let $\mathbf{x},\mathbf{y}\in \mathbb{S}^{d-1}$ and suppose $\mathbf{z}\in \mathbb{R}^d$ is such that $\|\mathbf{x}+\mathbf{z}\|_2,\|\mathbf{y}+\mathbf{z}\|_2\geq 1+\epsilon$ for some $\epsilon\geq 0$. Then $\|P_{\mathbb{S}^{d-1}}(\mathbf{x}+\mathbf{z})-P_{\mathbb{S}^{d-1}}(\mathbf{y}+\mathbf{z})\|_2\leq \frac{\|\mathbf{x}-\mathbf{y}\|_2}{1+\epsilon}.$
\end{lemma}
\begin{proof}
We may assume $\mathbf{x}\neq \mathbf{y}$, as otherwise the claim is trivial. Consider the arrangement of vectors $\mathbf{x}+\mathbf{z},\mathbf{y}+\mathbf{z},\mathbf{x}-\mathbf{y}$ in the plane spanned by $\{\mathbf{x}+\mathbf{z},\mathbf{y}+\mathbf{z}\}$ forming a triangle with a vertex at the origin and adjacent sides $\mathbf{x}+\mathbf{z},\mathbf{y}+\mathbf{z}$. By the assumption that these vectors have length at least $1+\epsilon$, scaling this triangle by a factor of $r=(1+\epsilon)^{-1}$ ensures that $r(\mathbf{x}+\mathbf{z})$ and $r(\mathbf{y}+\mathbf{z})$ continue to have at least unit norm, and the distance between them is at most $r(\|\mathbf{x}-\mathbf{y}\|_2)$. As projection onto $\mathbb{S}^{d-1}$ is a contraction in Euclidean distance for vectors of length at least $1$, it follows that
\begin{equation*}
    \|P_{\mathbb{S}^{d-1}}(\mathbf{x}+\mathbf{z})-P_{\mathbb{S}^{d-1}}(\mathbf{y}+\mathbf{z})\|_2=\|P_{\mathbb{S}^{d-1}}(r(\mathbf{x}+\mathbf{z}))-P_{\mathbb{S}^{d-1}}(r(\mathbf{y}+\mathbf{z}))\|_2\leq r\|\mathbf{x}-\mathbf{y}\|_2.
\end{equation*}
\end{proof}

Next, we show that with some constant probability, a random vector drawn from $\mathcal{D}$ will not split two vectors that form an acute angle, and that the probability of splitting at all tends to zero as the distance tends to zero.
\begin{lemma}
\label{lem:goodev1}
There exists constants $\beta,\gamma>0$ depending on $d$ and the measure of equivalence such that the following holds: suppose that $\mathbf{X}_0=(\mathbf{X}_0^{(1)},\mathbf{X}_0^{(2)})$ satisfies $\langle \mathbf{X}_0^{(1)},\mathbf{X}_0^{(2)}\rangle\geq 0$. Then the probability that $\xi_0$ satisfies:
\begin{enumerate}
    \item $\text{sgn}(\langle \mathbf{X}_0^{(1)},\xi_0\rangle) = \text{sgn}(\langle \mathbf{X}_0^{(2)},\xi_0\rangle)$, and
    
    \item $\vert \langle \mathbf{X}_0^{(i)},\xi_0\rangle\vert\geq \gamma$ for $i=1,2$,
\end{enumerate}
is at least $\beta$.
\end{lemma}
\begin{proof}
We first show this when $\xi_0$ is drawn uniformly from the sphere and then simply change constants when moving to any measure that is equivalent to Haar measure. But this is clear: the probability of $\xi_0$ satisfies the first property is at least $1/2$ under these assumptions, as the direction of $\xi_0$ in the plane spanned by $\mathbf{X}_0^{(1)},\mathbf{X}_0^{(2)}$ is itself uniform and using the acuteness of the two vectors. Moreover, the distribution of $\vert \langle \xi_0,\mathbf{z} \rangle\vert$ does not depend on $\mathbf{z}$ by uniformity, so we may choose $\gamma>0$ small enough so that the probability of the second property is at least $7/8$ for any fixed $\mathbf{z}$. By a union bound, it follows that the probability $\xi_0$ has the desired properties is at least $1/4$ under Haar measure. Under the true, equivalent distribution, the probability is thus at least $\beta:=C/4>0$.
\end{proof}

\begin{lemma}
\label{lem:bad1}
Under the assumptions and notation of \Cref{lem:goodev1}, the probability that $\text{sgn}(\langle \mathbf{X}_0^{(1)},\xi_0\rangle) \neq \text{sgn}(\langle \mathbf{X}_0^{(2)},\xi_0\rangle)$ is at most $O(\|\mathbf{X}_0^{(1)}-\mathbf{X}_0^{(2)}\|_2)$ and $O(\angle(\mathbf{X}_0^{(1)},\mathbf{X}_0^{(2)}))$, where the implicit constant depends only on $d$ and the measure of equivalence.
\end{lemma}
\begin{proof}
We first show this for Haar measure. By an analogous argument, the set of vectors with the desired property have directions lie in a band of width $O(\|\mathbf{X}_0^{(1)}-\mathbf{X}_0^{(2)}\|_2)$ in the plane spanned by $\mathbf{X}_0^{(1)},\mathbf{X}_0^{(2)}$, and therefore has probability at most $O(\|\mathbf{X}_0^{(1)}-\mathbf{X}_0^{(2)}\|_2)$ and $O(\angle(\mathbf{X}_0^{(1)},\mathbf{X}_0^{(2)}))$ by uniformity. For the true distribution, this can increase by a factor of at most $C'$, which we may absorb into the implicit constant.
\end{proof}

With this result in hand, we turn to the proof of the main theorem:
\begin{proof}[Proof of \Cref{thm:shjmr}]
We begin with a series of reductions that simplifies the problem. First, because these dynamics are oblivious, we observe that by \Cref{lem:2suff} it suffices to consider the case $n=2$ with an arbitrary starting vector $\mathbf{X}_0=(\mathbf{X}_0^{(1)},\mathbf{X}_0^{(2)})$. Next, by sign-invariance of these dynamics, we may assume that $\langle \mathbf{X}_0^{(1)},\mathbf{X}_0^{(2)}\rangle\geq 0$ by possibly flipping the sign of one of the vectors and noting that both the dynamics and the set $P$ are invariant under these sign changes. Note that this implies that the two starting vectors form an acute angle. Finally, by \Cref{thm:levy}, it suffices to show that the probability that $\|\mathbf{X}_t^{(1)}-\mathbf{X}_t^{(2)}\|_2\to 0$ is bounded below by some nonzero constant $c>0$, uniformly over the choice of starting vector (though assuming nonnegative inner product).

We use the notation of \Cref{lem:goodev1}. Note that the good event of \Cref{lem:goodev1} and the bad event of \Cref{lem:bad1} are disjoint, though not mutually exhaustive. We claim that by the craps principle, the probability of encountering a random update $\xi_t$ satisfying the good event before the bad event is at least $\beta/(\beta+O(\|\mathbf{X}_0^{(1)}-\mathbf{X}_0^{(2)}\|_2))$. Indeed, while an update need not satisfy either event, on the complement of the bad event, $\|\mathbf{X}_t^{(1)}-\mathbf{X}_t^{(2)}\|_2$ is nonincreasing so long as the bad event does not occur as the unnormalized lengths increase by sign-invariance with respect to $\xi_t$ (so that we may assume both signs are positive) and contractions decrease distances. The claim then follows from the crap's principle, \Cref{lem:goodev1}, and \Cref{lem:bad1}.

Moreover, on this event, the distance between $\mathbf{X}_{t+1}^{(1)}$ and $\mathbf{X}_{t+1}^{(2)}$ decreases by a factor of $(1+\epsilon)^{-1}$ where $\epsilon=\Omega(\eta\gamma)$ by \Cref{lem:cont} (using $\mathbf{z}=\pm \eta\cdot \xi_t$) and the fact that the inner products are bounded below on the good event. It follows that on this event, the new distance between vectors is at most $\|\mathbf{X}_0^{(1)}-\mathbf{X}_0^{(2)}\|_2/(1+\epsilon)$. By the strong Markov property, we may iterate this argument to show that the probability of the good event occurring before the bad event is now at least $\beta/(\beta+O(\|\mathbf{X}_0^{(1)}-\mathbf{X}_0^{(2)}\|_2/(1+\epsilon)))\geq 1-O\left(\frac{\|\mathbf{X}_0^{(1)}-\mathbf{X}_0^{(2)}\|_2}{1+\epsilon}\right)$ where we absorb the constant $\beta$. It follows that the probability that the good event occurs infinitely often without the bad event is at least
\begin{equation*}
    \prod_{k=0}^{\infty} \left(1-O\left(\frac{\|\mathbf{X}_0^{(1)}-\mathbf{X}_0^{(2)}\|_2}{(1+\epsilon)^k}\right)\right)\geq \prod_{k=0}^{\infty} \left(1-O\left(\frac{1}{(1+\epsilon)^k}\right)\right),
\end{equation*}
where we simply upper bound $\|\mathbf{X}_0^{(1)}-\mathbf{X}_0^{(2)}\|_2\leq 2$ in the inequality. Note that if this occurs, then \Cref{lem:cont} show that this implies that $\|\mathbf{X}_t^{(1)}-\mathbf{X}_t^{(1)}\|_2\to 0$ as the distance is nonincreasing and geometrically decreases infinitely often. From standard analysis, 
\begin{equation*}
    \prod_{k=0}^{\infty} \left(1-O\left(\frac{1}{(1+\epsilon)^k}\right)\right)>0 \iff \sum_{k=0}^{\infty}O\left(\frac{1}{(1+\epsilon)^k}\right)<\infty,
\end{equation*}
and the latter is clearly true as a geometric series. Moreover, these lower bounds are uniform over the value of $\|\mathbf{X}_0^{(1)}-\mathbf{X}_0^{(2)}\|_2$. By the reductions above, this completes the proof.
\end{proof}

\subsection{Strong Polarization in the Party Model}

We now turn to proving strong polarization in the party model. One complicating factor is that these dynamics are \emph{not} oblivious, unlike the other models where strong polarization is known. Therefore, we have to reason about multiple vectors acting on each other at the same time. 

To set up the formal statement of the result, we need the following definition: for any given set of nonnegative coefficients $\bm{\eta}$, let $A=A(\bm{\eta})$ be the (directed) $n\times n$ adjacency matrix defined by
\begin{equation*}
    A_{ij}=\begin{cases}
1 & \text{if $\eta^{(i)}_j>0$}\\
0 & \text{otherwise}.
    \end{cases}
\end{equation*}
In other words, we consider the directed graph with a directed edge from $i$ to $j$ if $j$ influences $i$. We say that $A$ is irreducible if $A^n>0$ componentwise. This is equivalent to the existence of a directed path between any two agents $i$ and $j$. Our main result of this section is as follows:

\begin{theorem}
\label{thm:party}
Suppose there are $n\geq 1$ agents in the party model where each $\xi_t$ is drawn i.i.d. from a distribution $\mathcal{D}$ that is equivalent to the uniform (Haar) measure $\mu$, i.e. there exists $C,C'>0$ such that for every measurable set $B\subseteq \mathbb{S}^{d-1}$, $C\mu(B)\leq \Pr(\xi_t\in B)\leq C'\mu(B)$. Moreover, suppose that $A=A(\bm{\eta})$ is irreducible. Then this system strongly polarizes from any choice of starting vector $\mathbf{X}_0$.
\end{theorem}

To prove this, we follow a similar high-level plan as that of \Cref{thm:shjmr}. By sign-invariance, we will assume via \Cref{lem:split} that each component of $\mathbf{X}_0$ lies on one side of a hyperplane with margin strictly bounded below by zero regardless of the individual configuration. We then construct a potential function that is equivalent to the maximum angle between any two agents with the property that, assuming the dynamics do not split up the vectors on the next $n$ iterations, it is guaranteed to decay by some factor strictly bounded above by $1$. Since, as we will again see, the probability that a random vector splits up any two components is essentially bounded by the maximum angle between any two components and this quantity is decreasing geometrically, it will follow that the probability of strong polarization is bounded below uniformly. We may then again conclude via \Cref{thm:levy} that the system strongly polarizes from any starting configuration. 

We now carry out this plan. First, we show that there always exists a signing of the starting configuration such that all vectors lie on one side of some hyperplane with nontrivial margin:
\begin{lemma}
\label{lem:split}
For all $d,n\geq 1$, there exists a constant $\lambda=\lambda(n,d)>0$ such that for any $\mathbf{x}_1,\ldots,\mathbf{x}_n\in \mathbb{S}^{d-1}$, there exists $\mathbf{z}\in \mathbb{S}^{d-1}$ such that $\vert \langle \mathbf{z},\mathbf{x}_i\rangle\vert\geq \lambda$ for all $i\in [n]$.
\end{lemma}
\begin{proof}
Simply choose $\lambda>0$ such that the probability of a random unit vector drawn from Haar measure does not satisfy the condition for a given $i$ is at most $1/(n+1)$ and then apply a union bound to conclude there exists such a vector. Note that this choice of $\lambda$ indeed depends on $d,n$, but not on the choice of vectors as the distribution of the inner product of a uniformly random vector on the sphere with a fixed vector does not depend on the identity of this fixed vector.
\end{proof}

Next, we proceed with several purely geometric results that will enable us to show contraction of a suitable potential function.
\begin{lemma}
\label{lem:conic}
Let $\mathbf{x}_1,\ldots,\mathbf{x}_n\in \mathbb{S}^{d-1}$ all lie strictly on one side of a hyperplane in $\mathbb{R}^d$. Define $\mathbf{y}$ by 
\begin{equation*}
    \mathbf{y}\in \arg\min\limits_{\mathbf{v}\in \mathbb{S}^{d-1}} \max_{j\in [n]} \angle (\mathbf{v},\mathbf{x}_j).
\end{equation*}
Note $\mathbf{y}$ exists as the objective function is continuous and the constraint set is compact. Let $H=\mathbf{y}^{\perp}$ be the orthogonal subspace to $\mathbf{y}$ and let $P_H$ be the orthogonal projection onto $H$. Then $\mathbf{0}\in \mathrm{conv}(P_H(\mathbf{x}_1),\ldots,P_H(\mathbf{x}_n))$.
\end{lemma}

\begin{proof}
Observe that the desired claim is implied by $\mathbf{y}\in \mathrm{cone}(\mathbf{x}_1,\ldots,\mathbf{x}_n)$. To see this, suppose that this holds: suppose there exists $\alpha_i\geq 0$ such that $\mathbf{y}=\sum_{i=1}^n \alpha_i \mathbf{x}_i$. Note that not all $\alpha_i$ are zero as $\mathbf{y}$ has unit norm by construction. Then applying $P_H$ to both sides, we deduce that 
\begin{equation*}
    \mathbf{0}=P_H(\mathbf{y})=\sum_{i=1}^n \alpha_i P_H(\mathbf{x}_i).
\end{equation*}
As $\sum_{i=1}^n \alpha_i>0$, we may renormalize these conefficients so that their sum is one to deduce that $\mathbf{0}\in \mathrm{conv}(P_H(\mathbf{x}_1),\ldots,P_H(\mathbf{x}_n))$. Moreover, note that $\mathbf{y}$ can equivalently be defined via $\mathbf{y}\in \arg\max\limits_{\mathbf{v}\in \mathbb{S}^{d-1}} \min_{j\in [n]} \langle \mathbf{y},\mathbf{x}_j\rangle$. This holds as the inner product is a monotonically decreasing function of angle on $[0,\pi]$ and so the optimization problems are equivalent. 

Therefore, suppose that $\mathbf{y}\not\in \mathrm{cone}(\mathbf{x}_1,\ldots,\mathbf{x}_n)$. By the separating hyperplane theorem, there exists a unit vector $\mathbf{u}$ such that $\langle \mathbf{u}, \mathbf{x}_i\rangle>0$ for all $i\in [n]$ but $\langle \mathbf{u},\mathbf{y}\rangle<0$. Write $\mathbb{R}^d$ in an orthogonal basis that includes $\mathbf{u}$. Then the coordinate with respect to $\mathbf{u}$ for each of the $\mathbf{x}_i$ is strictly positive, while the coordinate for $\mathbf{y}$ is strictly negative. Therefore, by reflecting $\mathbf{y}$ about $\mathbf{u}$, we obtain unit $\tilde{\mathbf{y}}$ such that $\langle \tilde{\mathbf{y}},\mathbf{x}_i\rangle>\langle \mathbf{y},\mathbf{x}_i\rangle $ for all $i\in [n]$, contradicting the optimality of $\mathbf{y}$. 
\end{proof}

\begin{lemma}
\label{lem:contract}
Let $\mathbf{z}_1,\ldots,\mathbf{z}_n$ satisfy $\|\mathbf{z}_i\|\leq 1$ for all $i\in [n]$ and $\mathbf{0}\in \mathrm{conv}(\mathbf{z}_1,\ldots,\mathbf{z}_n)$. Then \begin{equation*}
    \left\|\frac{1}{n}\sum_{i=1}^n \mathbf{z}_i\right\|_2\leq 1-\frac{1}{n}.
\end{equation*}
\end{lemma}
\begin{proof}
We consider two cases:
\begin{enumerate}
    \item First, suppose that $\mathbf{0}\in \{\mathbf{z}_1,\ldots,\mathbf{z}_n\}$; without loss of generality, suppose $\mathbf{z}_1=\mathbf{0}$. Then 
    \begin{equation*}
        \left\|\frac{1}{n}\sum_{i=1}^n \mathbf{z}_i\right\|_2=\left\|\frac{1}{n}\sum_{i=2}^n \mathbf{z}_i\right\|_2
    \end{equation*}
    As $\mathbf{0}$ is trivially in the convex hull no matter the choice of $\mathbf{z}_2,\ldots,\mathbf{z}_n$, the convexity of the Euclidean norm implies that this latter quantity is optimized for $\mathbf{z}_2=\ldots=\mathbf{z}_n=\mathbf{u}$ for some unit vector $\mathbf{u}$. This has norm $1-\frac{1}{n}$ as needed.
    
    \item Now suppose that $\mathbf{0}\not\in \{\mathbf{z}_1,\ldots,\mathbf{z}_n\}$. We claim that there exists some $j\in [n]$ such that setting $\mathbf{z}_j=\mathbf{0}$ does not decrease the desired quantity, thereby reducing to the previous case. 
    
    To see this, suppose that this is false. In particular, for every $j\in [n]$,
    \begin{equation*}
        \left\|\frac{1}{n}\sum_{i=1}^n \mathbf{z}_i\right\|_2>\left\|\frac{1}{n}\sum_{i\neq j} \mathbf{z}_i\right\|_2
 \iff
        \left\langle \mathbf{z}_j,\sum_{i\neq j}\mathbf{z}_i\right\rangle >0.
    \end{equation*}
   This in turn clearly implies that
    \begin{equation}
    \label{eq:nonneg}
        \left\langle \mathbf{z}_j,\sum_{i=1 }^n\mathbf{z}_i\right\rangle >0.
    \end{equation}
    By assumption, we may write $\mathbf{0}=\sum_{i=1}^n \alpha_i \mathbf{z}_i$ for some nonnegative scalars summing to $1$. Multiplying \Cref{eq:nonneg} by $\alpha_j$ for each $j\in [n]$ and summing, we obtain
    \begin{equation*}
        0=\left\langle \mathbf{0},\sum_{i=1 }^n\mathbf{z}_i\right\rangle = \left\langle \sum_{j=1}^n\alpha_j \mathbf{z}_j,\sum_{i=1 }^n\mathbf{z}_i\right\rangle>0,
    \end{equation*}
    a contradiction.
\end{enumerate}
\end{proof}

Next, we show that the irreducibility of $A$ implies that there is geometric decay in the minimal angle in each $n$ steps that the dynamics do not split the vectors. We start with the following crude, but intuitive lemma:

\begin{lemma}
\label{lem:stochmat}
Let $\mathbf{x}_0=(\mathbf{x}_0^{(1)},\ldots,\mathbf{x}_0^{(n)})$ be any set of vectors all lying strictly on one side of a hyperplane with $\lambda=\lambda(n,d)>0$ margin (i.e. there exists unit $\mathbf{v}$ satisfying $\langle \mathbf{v},\mathbf{x}_0^{(i)}\rangle>\lambda$ for all $i\in [n]$), and consider the update rule given by \Cref{eq:party} where the second summand is empty. 

Suppose that we apply this update rule $n$ times to obtain $\mathbf{x}_n$. Then, there exists $\epsilon=\epsilon(n,d,\bm{\eta})>0$ such that $
    \mathbf{x}_n \propto M\mathbf{x}_0
$
for some row-stochastic matrix $M$ depending on $\mathbf{x}_0$ satisfying
\begin{equation}
    M = \epsilon \Pi+(1-\epsilon)Q,
\end{equation}
where every entry of $\Pi$ is equal to $1/n$ and $Q$ is some arbitrary stochastic matrix.
\end{lemma}
\begin{remark}
In the above matrix equation, we interpret $\mathbf{x}_t$ as a $n\times d$ matrix where the $i$th row is $\mathbf{x}_t^{(i)}$. Moreover, the point of the lemma is that $\epsilon$ depends on $n,d,\bm{\eta}$, but \emph{not} on the starting configuration.
\end{remark}
\begin{proof}
We show the following claim by induction: for each $1\leq k\leq n$, there exists $\epsilon_k=\epsilon(n,d,\bm{\eta},k)$ such that $\mathbf{x}_k\propto M_k \mathbf{x}_0$ where $M_k$ is a row-stochastic matrix such that $(M_k)_{ij}>\epsilon_k$ if $j\in \Gamma^k(i)$, where $\Gamma^k(i)$ is the set of nodes reachable from $i$ in $k$ steps in the directed graph induced by $A=A(\bm{\eta})$ above. This clearly implies the claim by setting $\epsilon=\epsilon_n$ noting that $\Gamma^k(i)\subseteq \Gamma^{k+1}(i)$ by the fact that this matrix has ones on the diagonal.

For the base case $k=1$, by definition (and absorbing the term $\mathbf{x}_0^{(i)}$ into the $\eta_i^{(i)}$ multiplier): $
    \mathbf{x}_1^{(i)}\propto \sum_{j=1}^n \eta_j^{(i)}\mathbf{x}_0^{(j)}.$
By dividing by $\eta^{(i)}\triangleq \sum_{j=1}^n \eta_j^{(i)}$, we clearly obtain the claim with $\epsilon_1 \triangleq \min_{i\in [n]}\min\limits_{j\in [n]:\eta_j^{(i)}>0} \eta_j^{(i)}/\eta^{(i)}$. Note that this is independent of $\mathbf{x}_0$.

Now suppose it holds for some $k\geq 1$ so that $\mathbf{x}_k \propto M_k\mathbf{x}_0$.
By applying the base case to $\mathbf{x}_{k+1}$ (noting that $\mathbf{x}_k$ still lies on the same side of the hyperplane by convexity with same margin),
\begin{equation*}
    \mathbf{x}_{k+1} \propto M_1'\mathbf{x}_k\propto (M_1'M_k)\mathbf{x}_0.
\end{equation*}
By the induction hypothesis, $(M_1')_{ij}\geq \epsilon(n,d,\bm{\eta},1)$ for all $j\in \Gamma(i)$, while $(M_k)_{ij}\geq \epsilon(n,d,\bm{\eta},k)$ for all $j\in \Gamma^k(i)$. For any $j\in \Gamma^{k+1}(i)$, there exists some $j'\in [n]$ such that $j'\in \Gamma^k(i)$ and $j\in \Gamma(j')$. From the definition of matrix multiplication, if $j\in \Gamma^{k+1}(i)$, it follows that $(M_1'M_k)_{i,j}\geq \epsilon_1\cdot \epsilon_k$. Letting $\epsilon_{k+1}\triangleq \epsilon_1\cdot \epsilon_k$ and noting the product of stochastic matrices is stochastic, the claim follows.
\end{proof}

Finally, we show that one can get the geometric rate of convergence with respect to a natural potential function.
\begin{lemma}
\label{lem:geom}
Let $\mathbf{x}_1,\ldots,\mathbf{x}_n\in \mathbb{S}^{d-1}$ all lie strictly on one side of a hyperplane in $\mathbb{R}^d$ with margin at least $\lambda=\lambda(n,d)>0$. Define $\Phi(\mathbf{x}_1,\ldots,\mathbf{x}_n)$ by 
\begin{equation}
\label{eq:optvec}
    \Phi(\mathbf{x}_1,\ldots,\mathbf{x}_n) = \min_{\mathbf{v}\in \mathbb{S}^{d-1}}\max_{i\in [n]} \angle (\mathbf{v},\mathbf{x}_i).
\end{equation}

Let $\mathbf{x}_1',\ldots,\mathbf{x}_n'$ be the updated vectors after $n$ iterations of the update rule in \Cref{eq:party} when the second summand is empty. Then there exists $c=c(n,d,\bm{\eta})>0$ such that
\begin{equation*}
    \Phi(\mathbf{x}'_1,\ldots,\mathbf{x}'_n)\leq (1-c)\Phi(\mathbf{x}_1,\ldots,\mathbf{x}_n).
\end{equation*}
\end{lemma}
\begin{proof}
Let $\mathbf{y}$ be an optimizer in \Cref{eq:optvec} and let $H$ be the orthogonal complement. Let $P_H$ be the corresponding orthogonal projection onto $H$ and $P_{\mathbf{y}}$ be the orthogonal projection onto $\mathbf{y}$. Note that for any vector $\mathbf{z}$, $\mathbf{z}=P_H\mathbf{z}+P_{\mathbf{y}}\mathbf{z}$. By \Cref{lem:stochmat}, we have $\mathbf{x}'\propto M\mathbf{x}$ for some stochastic matrix $M$ such that $M=\epsilon \Pi+(1-\epsilon)Q$.

Now, observe that by the definition of $\Phi$ and from elementary geometry, $\|P_H(\mathbf{x}_i)\|_2\leq \sin(\Phi(\mathbf{x}_1,\ldots,\mathbf{x}_1))$ while $\langle \mathbf{y},\mathbf{x}_i\rangle=\langle \mathbf{y},P_{\mathbf{y}}\mathbf{x}_i\rangle\geq \cos(\Phi(\mathbf{x}_1,\ldots,\mathbf{x}_1))$ for all $i\in [n]$. Now by definition,
\begin{equation*}
    \Phi(\mathbf{x}'_1,\ldots,\mathbf{x}'_n)\leq \max_{i\in [n]} \angle (\mathbf{y},\mathbf{x}'_i) = \max_{i\in [n]} \angle (\mathbf{y},(M\mathbf{x})_i),
\end{equation*}
as angles do not change under positive scaling.

Therefore, we consider the (unnormalized) set of vectors $(M\mathbf{x})_i$. By linearity and the fact $M$ is row-stochastic, we still have $\|P_{\mathbf{y}}(M\mathbf{x})_i\|_2=\langle \mathbf{y}, (M\mathbf{x})_i\rangle\geq \cos(\Phi(\mathbf{x}_1,\ldots,\mathbf{x}_n))$. On the other hand,
\begin{equation}
\label{eq:sin1}
    \|P_H(M\mathbf{x})_i\|=\|\epsilon P_H( \overline{\mathbf{x}}) + (1-\epsilon)P_H \tilde{\mathbf{x}}_i\|,
\end{equation}
where $\overline{\mathbf{x}}=\frac{1}{n}\sum_{i=1}^n \mathbf{x}_i$ and $\tilde{\mathbf{x}}_i$ is some arbitrary convex combination. We thus have by linearity
\begin{equation}
\label{eq:sin2}
    \|P_H \tilde{\mathbf{x}}_i\|\leq \sin(\Phi(\mathbf{x}_1,\ldots,\mathbf{x}_n)).
\end{equation}
Moreover, again by linearity, $P_H\overline{\mathbf{x}}=\frac{1}{n}\sum_{i=1}^n P_H\mathbf{x}_i$. Recall that by \Cref{lem:conic}, $\mathbf{0}\in \mathrm{conv}(P_H\mathbf{x}_1,\ldots, P_H\mathbf{x}_n)$ and therefore by scaling and applying \Cref{lem:contract} with $\mathbf{z}_i=P_H\mathbf{x}_i$,
\begin{equation}
\label{eq:sin3}
    \|P_H\overline{\mathbf{x}}\|_2\leq \left(1-\frac{1}{n}\right)\sin(\Phi(\mathbf{x}_1,\ldots,\mathbf{x}_n)).
\end{equation}

Putting it all together, we now have by monotonicity of the function $\tan(\alpha)$ on $[0,\pi/2)$ and the triangle inequality with \Cref{eq:sin1,eq:sin2,eq:sin3} that
\begin{align*}
    \tan\left(\Phi(\mathbf{x}'_1,\ldots,\mathbf{x}'_n)\right)&\leq \max_{i\in [n]}\frac{\|P_H(M\mathbf{x})_i\|_2}{\|P_{\mathbf{y}}(M\mathbf{x})_i\|}\\
    &\leq \frac{\epsilon \left(1-\frac{1}{n}\right)\sin(\Phi(\mathbf{x}_1,\ldots,\mathbf{x}_n))+(1-\epsilon)\sin(\Phi(\mathbf{x}_1,\ldots,\mathbf{x}_n))}{\cos(\Phi(\mathbf{x}_1,\ldots,\mathbf{x}_n))}\\
    &=\left(1-\frac{\epsilon}{n}\right)\tan(\Phi(\mathbf{x}_1,\ldots,\mathbf{x}_n)).
\end{align*}
We may assume that $\Phi(\mathbf{x}_1,\ldots,\mathbf{x}_n)<\pi/2-\delta$ for some small enough $\delta=\delta(n,d)>0$ by the assumption that $\mathbf{x}_1,\ldots,\mathbf{x}_n$ lie on one side of a hyperplane with margin $\lambda=\lambda(n,d)>0$. The function $\tan(\cdot)$ has derivative bounded between $1$ and some constant depending on $n,d$ on the interval $[0,\pi/2-\delta]$ and thus by \Cref{lem:derivs}, the above geometric decay of $\tan\circ \Phi$ implies that $
    \Phi(\mathbf{x}'_1,\ldots,\mathbf{x}'_n)\leq (1-c)\cdot \Phi(\mathbf{x}_1,\ldots,\mathbf{x}_n)$
for some constant $c=c(n,d,\bm{\eta})>0$.
\end{proof}

\begin{lemma}
\label{lem:derivs}
Let $f:[0,a]\to [0,b]$ be a differentiable function such that $f(0)=0$ and $1\leq f'\leq K$ for some constant $K$. If $f(x)\leq (1-c)f(y)$, then $x\leq (1-c')y$ for $c'=c/2K$. 
\end{lemma}
\begin{proof}
Observe from the Mean Value Theorem, the assumption, and the fact $f(y)\geq f(x)\geq x$ from the derivative condition that 
\begin{equation*}
    y-x\geq \frac{f(y)-f(x)}{K}\geq \frac{cf(y)}{K}\geq \frac{cx}{K}.
\end{equation*}
It immediately follows that $y\geq \left(1+\frac{c}{K}\right)x$. As $(1+z)^{-1}\leq 1-z/2$ for $z\leq 1$, we obtain the claim.
\end{proof}

Finally, we can return to the proof of \Cref{thm:party}:
\begin{proof}[Proof of \Cref{thm:party}]
By sign-invariance and \Cref{lem:split}, we may assume that each component of $\mathbf{X}_0=(\mathbf{X}_0^{(1)},\ldots, \mathbf{X}_0^{(n)})$ lies strictly on one side of a hyperplane in $\mathbb{R}^d$ with margin at least $\lambda=\lambda(n,d)>0$; this follows because we may sign the starting configuration arbitrarily, run the dynamics, and then undo the signing by sign-invariance without affecting the polarization properties of the dynamics. We now show that with some constant probability (independent of $\mathbf{X}_0$), the dynamics monotonically decrease $\Phi(\mathbf{X}_t)$ to zero. Moreover, it is easy to see by the triangle inequality that for any vectors $\mathbf{x}_1,\ldots,\mathbf{x}_n$,
\begin{equation}
\label{eq:phiequiv}
    \Phi(\mathbf{x}_1,\ldots,\mathbf{x}_n)\leq \max_{i,j} \angle (\mathbf{x}_i,\mathbf{x}_j)\leq 2\Phi(\mathbf{x}_1,\ldots,\mathbf{x}_n).
\end{equation}
Therefore, $\Phi(\mathbf{X}_t)\to 0$ implies the same of the maximum angle between components and thus implies polarization. 

Observe that if the dynamics do not split the vectors on any of the next $n$ iterations, \Cref{lem:geom} implies that $\Phi(\mathbf{X}_n)\leq (1-c)\Phi(\mathbf{X}_0)$ for some constant $c>0$ independent of $\mathbf{X}_0$. This occurs with at least some nonzero constant probability $\delta$ (again, depending on $n,d,\bm{\eta}$, but not on $\mathbf{X}_0$) by the margin condition, \Cref{lem:bad1}, and the equivalence with the standard Haar measure. 

Now, recall from \Cref{lem:bad1} and a union bound that the probability that the dynamics split up any given $\mathbf{x_1},\ldots,\mathbf{x}_n$ that all lie strictly on one side of a hyperplane with is $O(\max_{i,j} \angle (\mathbf{x}_i,\mathbf{x}_j))=O(\Phi(\mathbf{x}_1,\ldots,\mathbf{x}_n))$ by \Cref{eq:phiequiv}, where the implicit constant depends on $n,d$ and the measure of equivalence. The probability that the sequence $\Phi(\mathbf{X}_{t\cdot n})\to 0$ is at least the probability that $\Phi(\mathbf{X}_{(t+1)\cdot n})\leq (1-c)\Phi(\mathbf{X}_{t\cdot n})$ for each $t\geq 0$. As we have shown that the probability that this fails decays geometrically, it follows that this latter probability is at least
\begin{equation}
\label{eq:problb}
    \prod_{t=0}^{\infty} \left(1-\Theta_{n,d,\bm{\eta}}\left((1-c)^t\right)\right)>c'
\end{equation}
for some constant $c'>0$ that does not depend on $\mathbf{X}_0$, where the inequality follows from the same standard analysis argument as in \Cref{thm:shjmr}. This immediately implies that strong polarization holds with constant probability on the restriction of the sequence to each $n$ steps. To extend this to the whole sequence, by inspecting the proof of \Cref{lem:geom}, it is easy to see that if the vectors are not split at some time $t+1$, then $\Phi(\mathbf{X}_{t+1})\leq \Phi(\mathbf{X}_t)$ for \emph{every} $t\geq 0$, not just on the subsequence (though we may not have strict contraction). In particular, the event that the dynamics never split up the vectors implies $\Phi(\mathbf{X}_t)\to 0$, and this holds with constant nonzero probability depending on just $n,d,\bm{\eta}$. By \Cref{eq:phiequiv}, \Cref{eq:problb}, and \Cref{thm:levy}, we conclude the result.
\end{proof}

\section{Weak Polarization and Markov Chains}
\label{sec:markov}

In the previous section, we showed that the strong polarization observed by HJMR extends to nontrivial variants of geometric opinion dynamics. However, we caution that it is \emph{not} generally true that just any geometric opinion model, even one that is continuous and $P$-invariant, will strongly polarize from an arbitrary configuration of starting opinions. As a trivial example, suppose that the updates are such that they simply apply a common orthogonal transformation to each vector at each time. This is clearly continuous and $P$-invariant, but obviously does not lead polarization except in very pathological examples.

While these trivial examples show that strong polarization does not necessarily hold in such models, it is natural to wonder if it is possible for other dynamics to satisfy other forms of polarization even if they do not satisfy strong polarization. In this section, we consider weaker forms of polarization than the strong form from \Cref{def:strongpolar} (restated as part (1) in the definition below):

\begin{definition}
Let $\mathbf{X}_t:=(\mathbf{X}_t^{(1)},\ldots, \mathbf{X}_t^{(n)})$ be a discrete-time Markov chain as given by \cref{eq:dynamics}.
\begin{enumerate}
    \item $\mathbf{X}_t$ \textbf{strongly polarizes} (from $\mathbf{X}_0$) if almost surely, $\rho(\mathbf{X}_t,P)\to 0$. 
    
    \item $\mathbf{X}_t$ \textbf{weakly polarizes} (from $\mathbf{X}_0$) if, for any fixed $\epsilon>0$, $\Pr(\rho(\mathbf{X}_t,P)\geq \epsilon) \overset{t\to \infty}{\to} 0.$
    
    \item $\mathbf{X}_t$ \textbf{weakly polarizes on average} (from $\mathbf{X}_0$) if, for any fixed $\epsilon>0$,
    \begin{equation*}
        \lim_{T\to \infty} \frac{\mathbb{E}\left[\sum_{t=1}^{T} I_{P_{\epsilon}}(\mathbf{X}_t)\right]}{T}=1,
    \end{equation*}
    where $I_{P_{\epsilon}}$ is the indicator of the set $P_{\epsilon}\triangleq \{\mathbf{z}\in (\mathbb{S}^{d-1})^n: \rho(\mathbf{z},P)\leq \epsilon\}$.
\end{enumerate}
\end{definition}

It is not difficult to see that these forms of convergence are listed in decreasing order of strength:

\begin{proposition}
\label{prop:forms}

Strong polarization implies weak polarization, which in turn implies weak polarization on average.
\end{proposition}
\begin{proof}
Suppose that $\mathbf{X}_t$ strongly polarizes. Then by definition, for any fixed $\epsilon>0$, there exists an almost surely finite (but random) $T$ such that for $t\geq T$, $\rho(\mathbf{X}_t,P)< \epsilon$. For any fixed $\delta>0$, there exists $K=K(\delta)$ large enough such that $\Pr(T\geq K)< \delta$. It follows that for $t\geq K$,
\begin{equation*}
    \Pr(\rho(\mathbf{X}_t,P)\geq \epsilon)\leq \Pr(T\geq K)< \delta.
\end{equation*}
As $\delta>0$ is arbitrary, it follows that the probability of being at least $\epsilon$-far from $P$ tends to zero. As $\epsilon>0$ is arbitrary, this gives weak polarization.

For the second implication, fix $\delta>0$ and let $K$ be large enough so that for all $t\geq K$, $\Pr(\rho(\mathbf{X}_t,P)\geq \epsilon/n)\leq \delta/2$. Then for all $T\geq 2K/\delta$, 
\begin{equation*}
    \frac{\mathbb{E}\left[\sum_{t=1}^{T} I_{P_{\epsilon}}(\mathbf{X}_t)\right]}{T}\geq \frac{\sum_{t=1}^{T} \Pr(\rho(\mathbf{X}_t,P)\leq \epsilon/n)}{T}\geq (1-\delta/2) - K/T
    \geq 1-\delta.
\end{equation*}
As this holds for all large enough $T$ and $\delta$ was arbitrary, the claim follows.
\end{proof}

\subsection{Polarization and Invariant Distributions}
With these weaker forms of polarization, we now proceed to show their connection to the theory of invariant distributions in general Markov chains. To begin, we will actually consider the simplest case of $n=1$ agent, where clearly all forms of polarization are trivial.

\begin{proposition}
\label{prop:invdis}
Consider $n=1$ agent updating opinions via dynamics that are continuous in the current state, i.e. the function $P_{\mathbb{S}^{d-1}}(\mathbf{x}+f_1(\mathbf{x},\xi))$ is continuous in $\mathbf{x}\in \mathbb{S}^d$ for every fixed realization of $\xi$. Then there exists an invariant distribution on $\mathbb{S}^{d-1}$ for these dynamics.

Moreover, assume 
\begin{enumerate}
    \item There exists some point $\mathbf{y}\in \mathbb{S}^{d-1}$ such that for any open neighborhood $U$ of $\mathbf{y}$, and any choice of starting point $x\in \mathbb{S}^{d-1}$ for the dynamics, there exists some $t=t(\mathbf{x},U)$ such that $\Pr_{\mathbf{x}}(\mathbf{X}^{(1)}_t\in U)>0$.
    
    \item The dynamics admit a continuous density with respect to the standard Haar measure on $\mathbb{S}^{d-1}$, i.e. distribution of $\mathbf{X}^{(1)}_1$ given $\mathbf{X}^{(1)}_0=\mathbf{x}$ has a density with respect to standard Haar measure.
\end{enumerate} 
Then the above invariant distribution is \emph{unique}.
\end{proposition}
\begin{proof}
From our continuity assumption, it follows that the dynamics are \emph{Feller}, i.e. for any continuous function $h:\mathbb{S}^{d-1}\to \mathbb{S}^{d-1}$, the map from $\mathbf{x}$ to $\mathbb{E}[h(\mathbf{x}')]$ where $\mathbf{x}'$ takes one step of the dynamics, remains continuous (Definition 2.36 of \cite{hairer2006ergodic}). Because $\mathbb{S}^{d-1}$ is compact and the dynamics are Feller, it follows that there exists at least one invariant measure (Corollary 4.18 of \cite{hairer2006ergodic}).

For the second part, the first additional assumption implies there exists an \emph{accessible} point $\mathbf{y}\in \mathbb{S}^{d-1}$, while the second implies that the dynamics satisfy the \emph{strong Feller} property (Definition 2.5 of \cite{hairer2010convergence}), i.e. the transition operator maps bounded measurable functions to continuous functions. From Corollary 2.7 of \cite{hairer2010convergence}, this implies that the dynamics have a unique invariant distribution.
\end{proof}

We now consider invariant distributions on \emph{the joint opinion vector} and then compare to the above results.

\begin{proposition}
Consider $n\geq 1$ agents updating opinions, and suppose the dynamics are jointly continuous for every fixed realization of $\xi$. Then there exists an invariant distribution for the vector $(\mathbf{X}^{(1)},\ldots,\mathbf{X}^{(n)})$ under these dynamics. Moreover, if the dynamics preserve $\sigma(D)$ for each $\sigma\in \{-1,1\}^n$, there exists such a distribution supported in $\bm{\sigma}(D)$ for every $\bm{\sigma}\in \{-1,1\}^n$.

If the dynamics are oblivious and satisfy the stronger assumptions of the previous proposition componentwise, then the $i$th marginal distribution of every such invariant distribution is unique.
\end{proposition}
\begin{proof}
The first part is identical to the proof of \Cref{prop:invdis} using the Feller property. The second claim holds from the assumed invariance of $\bm{\sigma}(D)$ and noting that the same argument holds when restricting to this invariant, compact set. The last claim follows from the obliviousness, as then the restriction of any invariant distribution to any component must yield an invariant distribution of the $n=1$ dynamics. This distribution is unique from \Cref{prop:invdis}.
\end{proof}
\begin{remark}
Note that by \Cref{prop:invdis}, it is \emph{not} true that there exists a unique invariant distribution for $n>1$ agents for most natural dynamics. The failure of the uniqueness of the invariant distribution arises because the joint dynamics typically will not have either accessible points nor satisfy the strong Feller property, as the set of possible updated values forms a lower-dimensional subset of $\prod_{i=1}^n \mathbb{S}^{d-1}$ and so cannot admit a density.
\end{remark}

With these results in hand, we can finally turn to the connection with polarization, which \emph{a priori} is a statement about whether dynamics starting at arbitrary points converge to the polarized set in some suitable sense. The connection is the following:

\begin{theorem}
\label{thm:wpoa}
Suppose that every invariant distribution of the joint dynamics with continuous updates has support contained in $P$. Then $\mathbf{X}_t$ weakly polarizes on average from any choice of starting vectors.
\end{theorem}
\begin{proof}
The proof of this result is implicit in the proof of the Krylov-Bogolubov theorem (Theorem 4.17 of \cite{hairer2006ergodic}). Fix any choice $\mathbf{X}_0$ of starting vectors for the joint dynamics. Define the family of probability measures for $T=1,2,\ldots$ by $Q_T(A) := \frac{1}{T}\sum_{t=0}^{T-1} \Pr(\mathbf{X}_t\in A).$
Then because the state space is compact, Prokhorov's Theorem (Theorem 4.15 of \cite{hairer2006ergodic}) asserts that this sequence of measures contains a weakly convergent subsequence $Q_{T_k}$ for $k=1,\ldots$, with weak limit we denote $\mu^*$; moreover, from the proof of the Krylov-Bogolubov theorem and the fact that we assumed the dynamics are Feller, $\mu^*$ is invariant.

Now, by our assumption, every invariant measure for these dyamics is concentrated on $P$. For any fixed $\epsilon>0$, let $P_{\epsilon}$ be as defined above. Because $\mu^*$ is supported on $P$ by assumption, $\mu^*(\partial P_{\epsilon})=0$, where $\partial A$ denotes the boundary of a set $A$. Recall that weak convergence of measures is equivalent to convergence on all continuity sets $A$ that satisfy $\mu^*(\partial A)=0$. It follows immediately that
\begin{equation*}
    Q_{T_k}(P_{\epsilon})=\frac{\sum_{t=0}^{T_k-1}\Pr(\mathbf{X}_t\in P_{\epsilon})}{T_k}\to \mu^*(P_{\epsilon})=1.
\end{equation*}

Moreover, this same argument holds for \emph{any} subsequence of $\{Q_T\}_{T=0}^{\infty}$, possibly with a different invariant measure which is nonetheless supported on $P$ by assumption. In particular, for any subsequence of the scalar sequence $\{Q_T(P_{\epsilon})\}_{T=1}^{\infty}$, there exists a further subsequence which converges to $1$. From standard analysis, any sequence such that every subsequence contains a further subsequence that converges to a fixed number $c$ must itself converge to $c$. Therefore, it follows that $
    Q_T(P_{\epsilon})\to 1,$
which by linearity of expectation is the statement of weak polarization on average.
\end{proof}

The converse always holds even when the dynamics are not continuous (though we do need to assume that there exists \emph{some} invariant distribution).

\begin{theorem}
\label{thm:supp}
Suppose that under the geometric opinion dynamics, there exists an invariant distribution $\mu^*$ that puts nontrivial mass on the set of nonpolarized vectors, i.e. $\mu^*(P^c)>0$. Then the opinion dynamics do not weakly converge on average. In particular, if the dynamics are continuous, then the dynamics weakly polarize on average if and only if every invariant distribution is supported on $P$.
\end{theorem}
\begin{proof}
The assumption and taking limits for $\epsilon\to 0$ with continuity of measure implies that there exists some $\epsilon>0$ such that $\mu^*(P_{\epsilon}^c)>0$. By invariance, this means that for every $t\geq 0$, $
    \Pr_{\mathbf{X}_0\sim \mu^*}(\mathbf{X}_t\in P_{\epsilon})=\mu^*(P_{\epsilon}^c)>0,$
which trivially implies that for every $T\geq 0$,
\begin{equation*}
    \mathbb{E}_{\mathbf{X}_0\sim \mu^*}\left[\mathbb{E}\left[\frac{\sum_{t=0}^{T-1} I_{P_{\epsilon}}}{T}\right]\right]=1-\mu^*(P_{\epsilon}^c),
\end{equation*}
where the inner expectation is over the sequence of random vectors $\{\xi_t\}_{t=1}^T$. By Fatou's lemma,
\begin{equation*}
    \mathbb{E}_{\mathbf{X}_0\sim \mu^*}\left[\liminf_{T\to \infty}\mathbb{E}\left[\frac{\sum_{t=0}^{T-1} I_{P_{\epsilon}}}{T}\right]\right]\leq \liminf_{T\to \infty}\mathbb{E}_{\mathbf{X}_0\sim \mu^*}\left[\mathbb{E}\left[\frac{\sum_{t=0}^{T-1} I_{P_{\epsilon}}}{T}\right]\right]= 1-\mu^*(P_{\epsilon}^c).
\end{equation*}
This implies that there exists some fixed $\mathbf{X}_0\in \prod_{i=1}^n \mathbb{S}^{d-1}$ such that 
\begin{equation*}
    \liminf_{T\to \infty}\mathbb{E}\left[\frac{\sum_{t=0}^{T-1} I_{P_{\epsilon}}}{T}\right]\leq 1-\mu^*(P_{\epsilon}^c)<1,
\end{equation*}
when the process is started at $\mathbf{X}_0$. By definition, it follows that the process does not weakly polarize on average with this choice of starting vector. The last claim then follows from combining the above with \Cref{thm:wpoa}.
\end{proof}

As an application, HJMR proved that for $d=2$ and $\xi_t$ is uniform on the unit circle, every finite set of vectors strongly polarizes. By \Cref{prop:forms}, this implies weak polarization on average. Because their dynamics are oblivious, admit a density for each single agent system, and every point is accessible, the unique invariant measure for the dynamics of a single agent is clearly uniform. \Cref{thm:supp} shows that every invariant measure on the \emph{joint} system must be supported on $P$.

To summarize, under very natural conditions, these geometric dynamics have invariant distributions in $P$, essentially by $P$-invariance of the dynamics. In the oblivious case, the projections of \emph{any} invariant distribution on each component will be unique under natural regularity assumptions and the only ambiguity is in how these distributions are coupled together. If one can show that these $P$-invariant measures are the \emph{only} such measures, or contrapositively that no invariant measure can put nontrivial mass on the complement of $P$, then the previous result immediately supplies weak polarization on average.

\subsection{HJMR Model and Weak Polarization}

Given the results in the previous sections, it is perhaps tempting to believe that under some necessary form of irreducibility, strong polarization will arise. In this section, we exhibit a model where \emph{weak polarization} on restricted subsets of initial configurations holds, but strong polarization provably does not. We do so in a very simple version of the original HJMR model by showing how the dynamics in this case relate to basic properties of simple random walks in $\mathbb{Z}$ and more generally, infinite balls-in-bins processes.

In this section, we will consider the HJMR dynamics where the update vector is drawn uniformly from a complete orthonormal set of basis vectors in $\mathbb{R}^d$; because the HJMR dynamics commute with orthogonal transformations, we will assume that these are the standard basis vectors. Towards showing that weak, but not strong, polarization holds, we need the following simple calculation showing that the random updates in this process are \emph{commutative}, greatly simplifying updating.

\begin{lemma}
\label{lem:update}
Let $\mathbf{v}\in \mathbb{S}^{d-1}$ be an arbitrary starting vector. Then for any sequence $\xi_1,\ldots,\xi_T$ of update vectors drawn from $\{\mathbf{e}_1,\ldots,\mathbf{e}_d\}$, the vector $\mathbf{z}$ obtained after $T$ steps of the dynamics given by \Cref{eq:hjmr} and starting at $\mathbf{v}$ satisfies
\begin{equation*}
    \mathbf{z}\propto \left((1+\eta)^{N^{(1)}_T}v_1,\ldots,(1+\eta)^{N^{(d)}_T}v_d)\right),
\end{equation*}
where $N^{(i)}_T$ is the number of times $\xi_t=\mathbf{e}_i$ up to time $T$.

In particular, with orthonormal update vectors, the vector obtained after $T$ steps of the dynamics depends only on the multiset $\{\xi_1,\ldots,\xi_T\}$ and not on the order they arrive.
\end{lemma}
\begin{proof}
Observe that the functional form of the update rule in \Cref{eq:hjmr} is homogeneous in the starting vector in that scaling the starting vector leads to the same update. As a consequence, one can perform the dynamics as stated by renormalizing after every step, or by applying the update rule without renormalizing until the end.

As a result of this observation, it suffices to show that for any vector $\mathbf{v}\in \mathbb{R}^{d}\setminus \{\mathbf{0}\}$, the effect of applying the update vector $\xi=\mathbf{e}_i$ is proportional to $(v_1,\ldots,(1+\eta)v_i,\ldots,v_d)$. This is sufficient as then we may iterate these updates without renormalizing until the only the end of the $T$ unnormalized updates to deduce commutativity. This is now immediate to see from \Cref{eq:hjmr}.
\end{proof}

Notice that the vector $(N^{(1)}_t,\ldots,N^{(d)}_t)$ is precisely a balls-in-bins process with $d$ bins and $t$ balls. The identity of \Cref{lem:update} shows that to understand the behavior of the process with orthonormal update vectors, we need to understand the sample path properties of an infinite balls-in-bins process. This is done in the following lemma:

\begin{lemma}
\label{lem:bib}
Consider an infinite balls-in-bins process with $d$ labeled bins $\{1,\ldots,d\}$. Let $N^{(i)}_t$ denote the number of balls in bin $i$ at time $t$. Then almost surely, the following holds:
\begin{enumerate}
    \item For each fixed $i\in [d]$, $N^{(i)}_t= \max_{j\in [d]}  N^{(j)}_t$ for infinitely many $t$.
    
    \item For each fixed $i,j\in [d]$, $N^{(i)}_t=N^{(j)}_t= \max_{k\in [d]}  N^{(k)}_t$ for infinitely many $t$.
\end{enumerate}
\end{lemma}
\begin{proof}
Recall from the Hewitt-Savage zero-one law that any event depending on a sequence of i.i.d. random variables that is invariant under finite permutations has probability $0$ or $1$ (Theorem 2.5.4 of \cite{durrett2019probability}). The first listed event evidently has this property. Clearly by the pigeonhole principle, there surely exists a (random) index $i^*\in [d]$ that is infinitely often the largest, and so by symmetry, the probability that any fixed $i\in[d]$ has this property is at least $1/d$ by a union bound. The Hewitt-Savage zero-one law then implies that the probability is one.

For the second claim, note that the previous claim and the fact that each $N^{(i)}_t$ is nondecreasing in time and changes by at most one in each round implies that infinitely often, there must be at least two indices of bins that are at least as large as the rest (if not, then with positive probability a single index is the unique largest element for all but finitely many times). Again by the pigeonhole principle and symmetry, each fixed pair $(i,j)$ has probability at least ${d \choose 2}^{-1}$ of having this property, hence by the Hewitt-Savage zero-one law has probability one.
\end{proof}

\begin{lemma}
\label{lem:bibfar}
Consider an infinite balls-in-bins process with $d$ labeled bins $\{1,\ldots,d\}$. Let $N^{(i)}_t$ denote the number of balls in bin $i$ at time $t$. Then for any fixed constant $K\in \mathbb{N}$, the probability that there exists two bins $i\neq j$ such that $\vert N^{(i)}_t-N^{(j)}_t\vert\leq K$ is $O(1/\sqrt{t})$, where the implicit constant depends on $K,d$, but not $t$.
\end{lemma}
\begin{proof}
Observe that the distribution of $N^{(i)}_t-N^{(j)}_t$ for any fixed $i\neq j$ is equal to that of a lazy, simple random walk on $\mathbb{Z}$ that updates with probability $2/d$. It is well-known (for instance, by the central limit theorem or Littlewood-Offord lemma) that for non-lazy simple random walks on $\mathbb{Z}$, the probability that the walk has magnitude at most some fixed constant $K$ at time $t$ is $O\left(\frac{1}{\sqrt{t}}\right)$, where the implicit constant depends on $K$. To account for laziness, note that with all but exponentially small probability, the number of updates to the $i$ or $j$th bin at time $t$ will exceed, say, $t/d$. We can then apply the result for non-lazy simple random walks and absorb the error term, just changing the implicit constant to depend on $d$ as well to account for the smaller number of steps. We may then apply a union bound over such pairs.
\end{proof}

We can now turn to a full analysis of the polarization behavior of this system.

\begin{theorem}
\label{thm:ortho}
Consider the HJMR dynamics in $\mathbb{S}^{d-1}$ given by \Cref{eq:hjmr}, where $\xi_t$ is drawn i.i.d. and uniformly from $\{\mathbf{e}_1,\ldots,\mathbf{e}_d\}$. Let $\mathbf{X}_0 = (\mathbf{X}_0^{(1)},\mathbf{X}_0^{(2)})\in (\mathbb{S}^{d-1})^2$ be such that the two vectors have equal supports
but are not equal up to sign. Then:
\begin{enumerate}
    \item The vector $\mathbf{X}_t$ weakly polarizes.
    \item With probability $1$, $\|\mathbf{X}_t^{(1)}-\mathbf{X}_t^{(2)}\|_2\not\to 0,2$. In particular, no nontrivial pair of starting vectors polarizes with nonzero probability.
\end{enumerate}
\end{theorem}

\begin{remark}
If the distribution of random updates is not uniform, but instead has strictly unequal probabilities associated to each basis vector, then it is easy to amend the argument to show that now \emph{strong} polarization holds between any two starting vectors with the same support, simply because both vectors will get concentrated about the basis vector in their support with the largest probability.

Moreover, note that there \emph{do} exist invariant distributions in this system, considered as $\prod_{i=1}^n \mathbb{S}^{d-1}$, that are not supported on $P$. However, this does not contradict \Cref{thm:supp} because \Cref{thm:ortho} only applies to vectors with equal supports, a full measure set. 
\end{remark}

\begin{proof}[Proof of \Cref{thm:ortho}]
Let $\mathbf{X}_0 = (\mathbf{X}_0^{(1)},\mathbf{X}_0^{(2)})\in (\mathbb{S}^{d-1})^2$ be as stated. Recall that to show weak polarization, we must show that for any fixed $\epsilon>0$, 
\begin{equation*}
    \Pr\left(\|\mathbf{X}_t^{(1)}-\mathbf{X}_t^{(2)}\|_2\land\|\mathbf{X}_t^{(1)}+\mathbf{X}_t^{(2)}\|_2>\epsilon\right)=o(1).
\end{equation*} 
Recall that from \Cref{lem:update}, the (unnormalized) $i$th component of each vector is amplified multiplicatively by a factor of $(1+\eta)^{N_t^{(i)}}$ before renormalization. From this, it is easy to see that for any fixed $\mathbf{X}_0$ and $\epsilon>0$, there exists some finite $K=K(\mathbf{X}_0,\epsilon)\in \mathbb{N}$ such that if $N_t^{(i)}\geq N_t^{(j)}+K$ for some $i\in [d]$ and all $j\neq i$, then both components of $\mathbf{X}_t$ are individually $\epsilon/2$-close to $\pm \mathbf{e}_i$. By the triangle inequality, it would then follow that $\mathbf{X}_t$ satisfies $\|\mathbf{X}_t^{(1)}-\mathbf{X}_t^{(2)}\|_2\land\|\mathbf{X}_t^{(1)}+\mathbf{X}_t^{(2)}\|_2\leq \epsilon$. Applying \Cref{lem:bibfar} with this value of $K$, we see that the probability of this event is $1-O\left(\frac{1}{\sqrt{t}}\right)$, where the constant depends only on $K,d$ and hence tends to $1$ with $t$ as needed. As $\epsilon>0$ was arbitrary, this proves the first claim.

For the second, we rely on the other properties we have already shown of the infinite balls-in-bins process. By assumption, as $\mathbf{X}_0^{(1)}$ and $\mathbf{X}_0^{(2)}$ are not equal up to sign, there exists a pair of indices such that the vectors restricted to these indices are not equal up to a constant. Without loss of generality, we assume that these indices are $1,2$. By \Cref{lem:bib}, there almost surely exist infinitely many times $t$ such that $N^{(1)}_t=N^{(2)}_t= \max_{k\in [d]}  N^{(k)}_t$. By the unnormalized form of the directions at each such time $t$ given by \Cref{lem:update}, it follows that at each such time, the first and second components of each vector remain proportional to their value at the beginning of this process but are now scaled by a factor that is at least $1$ after normalization. In particular, at each such time $t$ where this holds, the distance between $\mathbf{X}_t^{(1)}$ and $\mathbf{X}_t^{(2)}$ exceeds some strictly positive constant. Because this occurs infinitely often almost surely, the claim follows.
\end{proof}

\section{Conclusions}
In this paper, we have attempted to more closely study the polarization phenomenon that arises in geometric models as considered by HJMR. We introduce different forms of polarization and consider how these relate to more general Markov principles. Finally, we consider these various forms of polarization in concrete models to extend the class of natural dynamics that exhibit this property. We hope that some of the techniques we introduce here to complement the original results of HJMR prove useful in understanding the probabilistic properties of these models.

However, the most general version of this question remains open: are there very general, \emph{easily verified}, conditions on the update rule that ensures that some form of polarization arises (perhaps under some sort of irreducibility assumption)? Are there other interesting models of geometric opinion formation under which strong polarization holds? Understanding which sorts of dynamics lead to polarization, and also a better sense of why, may shed significant insight into the theoretical underpinnings of this phenomenon.

% In the interest of anonymization, please do not include acknowledgements in your submission.
%
\section{Acknowledgments}
We thank Ron Peretz for graciously suggesting the alternate proof of \Cref{thm:party} that appears here, which both simplified and strengthened the original analysis.

% Bibliography
\bibliographystyle{alpha}
\bibliography{Polarization.bib}

% Appendix
\end{document}